\documentclass[conference]{IEEEtran}
\IEEEoverridecommandlockouts

\usepackage[utf8]{inputenc} %
\usepackage[T1]{fontenc}

\usepackage{amsmath,amssymb,amsfonts}
\usepackage{algorithmic}
\usepackage{graphicx}
\usepackage{textcomp}
\def\BibTeX{{\rm B\kern-.05em{\sc i\kern-.025em b}\kern-.08em
    T\kern-.1667em\lower.7ex\hbox{E}\kern-.125emX}}

\usepackage{ifthen,soul}

\newboolean{review}
\setboolean{review}{false}

\newcommand{\setreviewremovetrue}[1]{\textcolor{red}{\st{#1}}}
\newcommand{\setreviewremovefalse}[1]{\unskip}

\newcommand{\setreviewaddtrue}{\textcolor{blue}}
\newcommand{\setreviewaddfalse}{\textcolor{black}}

\newcommand{\reviewremove}{\ifthenelse{\boolean{review}}
	{\setreviewremovetrue}             
	{\setreviewremovefalse}}

\newcommand{\reviewadd}{\ifthenelse{\boolean{review}}
	{\setreviewaddtrue}
	{\setreviewaddfalse}}

\usepackage{amsthm}
\newtheorem{theorem}{Theorem}

\newtheorem{definition}{Definition}

\newtheorem{lemma}{Lemma}

\usepackage[dvipsnames]{xcolor}
\usepackage[shadow,colorinlistoftodos]{todonotes}

\usepackage{biblatex} 
\usepackage{subfig}
\graphicspath{{figures/}}

\usepackage{tikz}
\usetikzlibrary{arrows,positioning,automata,graphs,calc,matrix,fit}
\tikzset{treenode/.style = {circle, inner sep=0pt, text centered,
		font=\sffamily\bfseries,draw=black,
		fill=white, text width=2em, very thick, align=center}}
\tikzset{channelnode/.style = {circle, inner sep=0pt, text centered,
		font=\sffamily\bfseries,draw=black,
		fill=white, text width=2em, very thick, align=center}}
\tikzset{channelnodelarge/.style = {circle, inner sep=0pt, text centered,
		font=\sffamily\bfseries,draw=black,
		fill=white, text width=3em, very thick, align=center}}
\tikzset{blanknode/.style = {}}

\newcommand{\ProtocolEncryptedInput}{1.1}
\newcommand{\ProtocolEncryptedPlayersRandoms}{1.2}
\newcommand{\ProtocolSequencePermutedRandomValues}{2.1}
\newcommand{\ProtocolSequenceEncryptedInputs}{2.2}
\newcommand{\ProtocolSequenceEncryptedInputsBlindings}{2.3}
\newcommand{\ProtocolServiceProvidersRandom}{2.4}
\newcommand{\ProtocolSequenceHashes}{2.5}
\newcommand{\ProtocolRandomPositionHash}{2.6}
\newcommand{\ProtocolEncryptedRerandomisedBlindedInput}{2.7}
\newcommand{\ProtocolEncryptedBlindedBlinding}{2.8}
\newcommand{\ProtocolEncryptedBlindingBlinding}{2.9}
\newcommand{\ProtocolSequenceRerandomisedPermutedInputs}{2.10}

\begin{document}

\title{Efficient Cloud-based Secret Shuffling via Homomorphic Encryption}

\author{\IEEEauthorblockN{Kilian Becher}
	\IEEEauthorblockA{\textit{Chair of Privacy and Data Security}\\
		\textit{TU Dresden}\\
		Dresden, Germany\\
		kilian.becher@tu-dresden.de}
	\and
	\IEEEauthorblockN{Thorsten Strufe}
	\IEEEauthorblockA{\textit{Chair of IT Security}\\
		\textit{Karlsruhe Institute of Technology}\\
		Karlsruhe, Germany\\
		thorsten.strufe@kit.edu}}

\maketitle
\begin{abstract}
When working with joint collections of confidential data from multiple sources, e.g.\reviewadd{,} in cloud-based multi-party computation scenarios, the\reviewremove{position of a data provider's input in the data collection} \reviewadd{ownership relation between data providers and their inputs} itself is confidential information. Protecting \reviewremove{their} \reviewadd{data providers'} privacy desires a function for secretly shuffling the \reviewremove{inputs in the} data collection. We present \reviewremove{an} \reviewadd{the first} efficient \reviewadd{secure multi-party computation} protocol for secret shuffling \reviewadd{in scenarios with a central server}. \reviewremove{It} \reviewadd{Based on a novel approach to random index distribution, our solution} enables the \reviewremove{re} randomization of the order of a sequence of encrypted data such that no observer can \reviewremove{determine the position of any} \reviewadd{map between} element\reviewadd{s} of the original sequence \reviewremove{in} \reviewadd{and} the shuffled sequence with probability better than guessing. \reviewremove{Our solution is a secure multi-party computation protocol that} \reviewadd{It} allows for shuffling data encrypted under a\reviewadd{n additively} homomorphic cryptosystem with\reviewremove{additive homomorphism. It has} constant round complexity and linear computational complexity\reviewremove{ per participant}. \reviewadd{Being a general-purpose protocol, it is of relevance for a variety of practical use cases.}
\end{abstract}


\section{Introduction}

\reviewadd{In an industrial context, security against semi-honest adversaries~\cite{Lin17} is a valid assumption as companies typically have a financial and legal interest in the correct execution of processes. Proactive misbehaviour or negligent data handling could lead to a loss of reputation or legal consequences, such as those imposed by the European Union's General Data Protection Regulation (GDPR)~\cite{GDPR16}.}

\reviewremove{We consider scenarios where multiple participants each contribute a potentially confidential input to some service provider anonymously. We refer to the data providers as players and require the service provider to be a single, central instance (see Fig. fig:RequirementsCommunicationServiceProviderModel). Anonymity implies that the ownership relation between players and their inputs must not be preserved.}

\reviewremove{This requirement applies to a variety of real-life scenarios, such as anonymous message boards where people send messages to a trusted service provider who publishes those messages anonymously in the sense that no user can identify the author of a post. To prevent any external observer from linking messages to their authors, people could send their messages to the trusted service provider in encrypted form and the service provider could decrypt a batch of messages, change their order, and publish the messages on the message board. Even though no external observer could link people to their messages, the service provider could. Thus, such a construction requires trust in the service provider, that is, the service provider acts as a trusted third party (TTP).}

\reviewremove{In a similar way, one could perform anonymous surveys, polls, and voting. People could send their confidential responses or votes to the trusted service provider in encrypted form, the service provider could decrypt and accumulate the inputs, and publish the results of the survey, poll, and voting, respectively. However, as before, this does not ensure anonymity towards the trusted service provider and, most importantly, does not prevent it from learning the confidential inputs.}

\reviewremove{Another vivid example for centralized computations on confidential inputs is cross-company benchmarking. In such a benchmark, multiple companies each provide key performance indicators (KPI), e.g., return on investment, to some central party. The central party computes statistical measures, e.g., mean and quartiles, based on the companies KPIs and returns these statistical measures to the companies. Given these benchmark results, the companies can compare their performance to the overall performance of their industry. Hence, benchmarking is a valuable measure for companies to determine potential for improvement. However, they have to trust the service provider not to leak their private KPIs to their competitors.}

\reviewremove{In [KER10], the authors present a cryptographic protocol for privacy-preserving benchmarking that does not require a trusted third party. To ensure privacy of the companies, KPIs are processed entirely in encrypted form. Players do not communicate directly with each other but only with the central service provider. 
After running the protocol, each player knows its own KPI, the resulting statistical measures, and what can be inferred from that. The service provider knows the resulting statistical measures and what can be inferred from that but does not learn the players' secret KPIs.
To provide rank-based statistical measures like the median, the benchmarking protocol of [KER10] performs privacy-preserving sorting by obliviously comparing all values to each other, causing quadratic computational and communication complexity.}

\reviewremove{This complexity can be reduced by applying sorting with fewer comparisons but more interaction based on sorting networks as described in [KER10]. However, this potentially leads to participants learning the rank of other participants' inputs in the sorted list (todo: Citation). Despite encryption, the untrusted service provider would learn which player provided the largest value, median value, etc. and could therefore infer information on each player's performance. As a countermeasure, secretly shuffling the data sequence before sorting it can vastly reduce this leakage. As pointed out in [BBS18] for privacy-preserving benchmarking, the effect of the computational complexity in such scenarios is much larger than the effect of the communication complexity. Therefore, we require low computational complexity, potentially at the price of slightly increasing communication overhead.}

\reviewremove{A secret shuffle could easily be performed by a trusted third party. However, this would just shift the trust assumption from the service provider to a trusted third party used for shuffling rather than loosening the overall trust assumption. In this paper, we present an efficient, general-purpose secure multi-party computation protocol MPC for shuffling encrypted inputs based on an additively homomorphic cryptosystem. Key element of our contribution is a novel approach to random index distribution. A secure multi-party computation protocol emulates a trusted party by jointly evaluating a function. Such a protocol is secure in the sense that participants only learn their own inputs, their outputs, and what can be inferred from that[Gol02]. Our protocol is secure against semi-honest adversaries.}

\reviewadd{To make well-informed business decisions, companies need to determine their strengths and weaknesses. One widely-used measure is cross-company benchmarking.
In cross-company benchmarking, companies compare their key performance indicators (KPI), e.g., return on investment, to those of other companies of the same industry. As results, they obtain statistical measures, such as quartiles and mean. To compute rank-based statistical measures like quartiles, sorting KPIs typically is an important aspect of benchmarking. However, as the companies' performances are confidential, no company should learn another company's KPIs. Instead, benchmark results should only help them determine how they perform relatively to their overall peer group.
To ensure that, benchmarks typically are performed by trusted third parties (TTP), neutral companies that take the companies' KPIs in plaintext and centrally compute the statistical measures.
However, using a TTP requires trust. On the one hand, companies need to trust that the TTP does not proactively abuse the companies' private KPIs. As described above, this is a valid assumption as the neutral party has a financial and legal interest in honest behavior. However, on the other hand, they need to trust that the TTP implements sufficient security measures that prevent data breaches.
This is an important drawback of the TTP approach as data breaches might cost companies their competitive advantage or reputation.}

\reviewadd{Alternatively, benchmarking could be performed via secure multi-party computation (MPC)~\cite{Gol02}. An MPC protocol emulates a TTP by having the parties, e.g., companies, jointly evaluate some public function, e.g., quartile computation, over their inputs. Most importantly, those inputs are kept private, e.g., processed in an encrypted form. Such a protocol is secure in the sense that parties only learn their own inputs, their outputs, and what can be inferred from that. Hence, confidential KPIs are protected from any internal and external observer, enabling privacy-preserving benchmarking. We restrict our considerations to MPC scenarios where $n$ parties each contribute confidential inputs and jointly evaluate the target function with a service provider. We refer to the data providers as players and require the service provider to be a single, central instance (see Fig. \ref{fig:RequirementsCommunicationServiceProviderModel}).}

\reviewadd{As the core of a benchmarking MPC protocol, encrypted KPIs need to be sorted according to their underlying plaintexts in a privacy-preserving fashion. This can be done via sorting networks in up to $n \log n$ comparisons orchestrated by a service provider as described in~\cite{KER10}.}
\reviewadd{However, this would cause the service provider to learn the order of the confidential KPIs, that is, how a particular company performs relatively to another particular company. Even if the service provider is assumed to not misuse this information proactively, a data breach could leak this confidential performance information.}

\reviewadd{To reduce the risk of benchmarks leaking confidential data and relative performance information, an efficient privacy-preserving benchmarking protocol based on MPC should ensure anonymity in the sense that no observer can infer ownership relations between companies and their encrypted KPIs. This can be done by secretly shuffling the encrypted KPIs prior to benchmarking.}
We refer to a secret shuffle as a function that randomizes the order of a sequence of encrypted inputs such that no observer can map elements in the original sequence to their corresponding elements in the shuffled sequence with probability better than guessing. Preventing such a mapping also implies a need for changing the ciphertexts without affecting the underlying plaintexts.

Besides privacy-preserving benchmarking\reviewremove{ as well as anonymous message boards, surveys, polls, and voting}, our protocol can be applied to any scenario where $n$ players send encrypted inputs to a central service provider, e.g., a cloud service, without it learning which player provided which input. \reviewadd{This includes use cases such as anonymous surveys, polls, and voting.} Before we present our shuffling protocol in Section~\ref{Protocol}, we introduce required definitions and preliminaries and give an overview of related work. In Sections~\ref{PrivacyProof} and~\ref{CorrectnessProof}, we prove input privacy and correctness, respectively, before we evaluate the complexity and performance of our protocol in Section~\ref{Evaluationn}.%

\section{Preliminaries}\label{Preliminaries}

We restrict our considerations to asymmetric cryptosystems, i.e.\reviewadd{,} a tuple $\mathcal{CS} = (G,E,D)$ consisting of three polynomial-time algorithms. The probabilistic key-generation algorithm $G$ takes as input a security parameter $\kappa$ and outputs a key pair $(pk, sk)$ consisting of a public encryption key $pk$ and a secret decryption key $sk$. The probabilistic encryption algorithm $E$ takes as input a plaintext $x \in \mathcal{M}$ and  $pk$ and outputs the ciphertext $y = E(x,pk) \in \mathcal{C}$. $\mathcal{M}$ and $\mathcal{C}$ denote the plaintext and ciphertext space, respectively. The decryption algorithm $D$ takes as input a ciphertext $y \in \mathcal{C}$ and $sk$ and outputs the plaintext $x = D(y,sk) \in \mathcal{M}$.
For simplification, we denote the encryption of $x \in \mathcal{M}_i$ under a cryptosystem $\mathcal{CS}_i=(G_i,E_i,D_i)$ for $pk_i$ by $y=E_i(x)$\reviewremove{. Similarly, we denote} \reviewadd{and} the decryption of $y \in \mathcal{C}_i$ for $sk_i$ by $x=D_i(y)$.

\begin{figure}[t!]
	\begin{center}
		\includegraphics[width=1.85in]{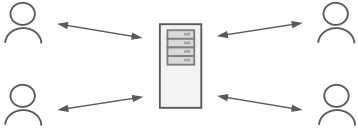}
	\end{center}
	\caption{Network of One Service Provider and Four Players}
	\label{fig:RequirementsCommunicationServiceProviderModel}
\end{figure}

Homomorphic cryptosystems, such as RSA~\cite{RSA78}, Paillier's~\cite{PAI99}, and BGV~\cite{BGV14}, allow for computations on ciphertexts. A cryptosystem $\mathcal{CS}$ is homomorphic if applying an operation $\circ$ to ciphertexts $E(x_1)$ and $E(x_2)$ yields the ciphertext $E(y)$ of the result $y = x_1 \ast x_2$ of a corresponding homomorphic operation $\ast$ applied to the plaintexts $x_1$ and $x_2$~\cite{KER08}. That is, $E(x_1) \circ E(x_2) = E(x_1 \ast x_2)$.
We restrict our considerations to cryptosystems with an additive homomorphism enabling addition of the underlying plaintexts as depicted in~\eqref{Eq:HEMult} and~\eqref{Eq:HEExp}, such as Paillier's cryptosystem~\cite{PAI99}.
\begin{equation}\label{Eq:HEMult}
	D(E(x_1) \cdot E(x_2)) = x_1 + x_2
\end{equation}
\begin{equation}\label{Eq:HEExp}
	D(E(x_1)^{x_2}) = x_1 \cdot x_2
\end{equation}
\reviewadd{That is, multiplication of ciphertexts encrypted under the same key $pk$ yields an encryption of the sum of the underlying plaintexts, encrypted under $pk$. This enables multiplication of an encrypted value by a plaintext value via exponentiation.}%

Paillier's cryptosystem allows for rerandomization~\cite{PAI99}. Given $pk$ and a ciphertext $E(x)$ of a plaintext $x$, rerandomization is an operation that \reviewremove{allows to} compute\reviewadd{s} a valid ciphertext $E^\prime(x)$ without decryption. With high probability, $E(x) \neq E^\prime(x)$ is ensured such that the output distributions of rerandomization and encryption are computationally indistinguishable~\cite{GG17}.
For Paillier's cryptosystem, it can be performed by multiplication with the encrypted identity element $0$ as depicted in~\eqref{Eq:Rerandomisation}~\cite{PAI99}.
\begin{equation}\label{Eq:Rerandomisation}
	E(x) \cdot E(0) = E(x+0) = E^\prime(x)
\end{equation}
A hash function $h(\cdot)$ is a function that, for arbitrarily long inputs $x$, computes outputs $h(x)$ of fixed length~\cite{KAL08}. It is easy to compute $h(x)$, called hash~\cite{KER10}. A hash function is cryptographic if it provides pre-image resistance and collision-resistance. \reviewremove{Pre-image resistance} \reviewadd{The former} guarantees that it is computationally infeasible to compute $x$ given only $h(x)$~\cite{KER10}. \reviewremove{Collision resistance} \reviewadd{The latter} ensures that it is computationally infeasible to find two hashes $h(x) = h(x^\prime)$ such that $x \neq x^\prime$~\cite{KAL08}.

We define a sequence $S$ as an enumeration of \reviewremove{objects} \reviewadd{elements} $s_i$ that are arranged in a particular order.\reviewremove{These $s_i$ are referred to as elements of $S$.} Multiple $s_i$ can have the same value. The number of $s_i$ in $S$ is referred to as its length.
We only use finite sequences of fixed length $n$ and denote them either by $S=(s_1,...,s_n)$ or by $S=(...,s_i,...)$ depending on whether we want to emphasize the elements' order or their form.
Given $S = (s_1,...,s_n)$, a random permutation $\pi : S \rightarrow S^\prime$ is a permutation that is chosen uniformly at random from the set of permutations on sequences of length $n$~\cite{KAL08}. It yields the permuted sequence $S^\prime = (...,s_i,... | s_i \in S)$ containing the same $n$ elements as $S$ but in a randomly permuted order.
We denote the position of \reviewremove{$s_i \in S$} \reviewadd{$s_i$} in $S^\prime$ permuted via $\pi$ by $\pi(s_i)$.

A function $\varphi(m)$ is called negligible in $m$ if for every polynomial $p(m)$ there is an $m_0 \in \mathbb{N}$ such that for any $m>m_0$, $\varphi(m) \leq \frac{1}{p(m)}$ applies~\cite{KAL08}. Let $\left\{X_m^1\right\}_{m \in \mathbb{N}}$ and $\left\{X_m^2\right\}_{m \in \mathbb{N}}$ be two sets of random variables. If for a probabilistic polynomial-time algorithm $A$ the advantage
\begin{equation}\label{eq:PreliminariesIndistinguishability}
\epsilon = \left| Pr\left[ A\left(X_m^1,1^m\right)=1\right] - Pr\left[ A\left(X_m^2,1^m\right)=1\right] \right|
\end{equation}
is negligible in $m$, the two sets are computationally indistinguishable~\cite{KAL08}.

In~\cite{Gro10}, a shuffle of a sequence of ciphertexts is defined as a sequence of different ciphertexts of the same plaintexts, arranged in a permuted order. We additionally require the permutation to be secret and define a secret shuffle as follows.
\begin{definition}[Secret Shuffle]\label{Def:SecureShuffle}
	Given a sequence of ciphertexts $X = (...,E(x_i),...)$ with $1 \leq i \leq n$. A secret shuffle $\mathcal{S}(\cdot)$ is a function that, for input $X$, yields a sequence $\mathcal{X} = (...,E^\prime(x_{\pi(i)}),...)$ such that the ciphertexts $E^\prime(x_1)\neq E(x_1),...,E^\prime(x_n)\neq E(x_n)$ have the same plaintexts $x_1,...,x_n$. The order of the elements in $\mathcal{X}$ is randomly permuted via a random permutation $\pi$. No participant can learn more than negligibly much information about $\pi$.
\end{definition}
\section{Related Work}\label{RelatedWork}

\subsection{Approaches with Additional (Neutral) Instances}\label{RelatedWorkAdditionalInstances}

In~\cite{Cha81}, Chaum introduces mix networks, a protocol that enables anonymity and unlinkability of messages to their senders at the cost of additional computational overhead. Mix networks involve a sequence of servers, called mixes, which receive a set of messages, shuffle\reviewadd{,} and forward them to the next mix~\cite{Gro10}. Unlinkability is guaranteed if at least one mix is honest~\cite{Gro10}. There are two kinds of shuffles: decryption and re-encryption shuffles~\cite{AW07}. In decryption shuffles, the messages are layered ciphertexts. Each mix removes one layer of encryption from each message and sorts the resulting plaintexts. In re-encryption shuffles, the \reviewremove{messages are} \reviewadd{mixes} rerandomize\reviewremove{d} and permute\reviewremove{d} \reviewadd{the messages} via a randomly chosen permutation. \reviewadd{A re-encryption mix network that ensures simplified key management based on universal re-encryption is given in~\cite{GJJS04}.}
In~\cite{Wik05}, the first mix network that is universally composable and efficient independent of the number of mixes is presented. The first efficient non-interactive zero-knowledge proof for proving that a mix shuffled correctly is proposed in~\cite{GL07}. A description of how the permutation used by a mix can be constructed by multiple parties %
is given in~\cite{DLOSS18}.

Unfortunately, mix networks cannot be used in the described scenario to ensure unlinkability between players and their inputs due to several drawbacks. Most importantly, mixes need to be provided by different, independent parties~\cite{PRT12}. This cannot be guaranteed in scenarios with a single, central service provider. \reviewadd{The same applies to Riffle~\cite{KLDF15}, a}\reviewremove{A}n alternative to mix networks\reviewremove{, called Riffle, is presented in [KLDF15]. It ensures sender and receiver anonymity. Similar to mix networks, it requires at least one additional independent server}.

\subsection{Approaches Based on Trusted Hardware}

Alternatively, unlinkability can be achieved by shuffling in trusted hardware, e.g.\reviewadd{,} Intel Software Guard Extensions (SGX)~\cite{SGX16}. Generating and applying the random permutation as well as rerandomization can be done inside trusted hardware on the service\reviewremove{ }\reviewadd{-}provider side. Such an approach is described in~\cite{DYD11} for database access pattern obfuscation. %
In \cite{DDCO17}, an approach with a trusted unit with limited private memory performing shuffling of encrypted data is given. An architecture for privacy-preserving online analysis of client data based on trusted hardware is presented in~\cite{BEMMRLRKTS17}.
In setups with a trusted CPU but no trusted memory, oblivious RAM (ORAM) can ensure that untrustworthy RAM does not leak confidential information~\cite{GO96}. %

However, \reviewremove{such an} \reviewadd{these} approach\reviewadd{es} impl\reviewremove{ies}\reviewadd{y} different trust assumptions \reviewadd{and relations such as trust in the manufacturing of the trusted hardware. Therefore, they are}\reviewremove{and therefore is} not suitable for our scenario with\reviewremove{ mutually} distrusting \reviewremove{companies} \reviewadd{participants}.

\subsection{Approaches Based on Secure Multi-party Computation}

One approach to secure multi-party computation is secret sharing. %
In~\cite{LWZ11}, three shuffling \reviewadd{MPC} protocols \reviewremove{ based on secret sharing} are proposed for the Sharemind secure computation platform, focusing on low communication and round complexity. In Sharemind, computation is done by three independent parties~\cite{BOG13}. This does not fit our scenario with a single, central service provider. Another shuffling protocol based on secret sharing is presented in~\cite{MSZ15}. %
However, it is designed for decentralized settings.

Secure multi-party computation can also be based on homomorphic encryption.
In~\cite{BS06},\reviewremove{ a} \reviewadd{such an MPC} protocol for shuffling data in a setting of multiple data providers and one data miner is proposed. %
It emulates a mix network in the sense that each data provider itself acts as a mix. Hence, it does not require independent mix servers. %
With its quadratic computational and linear round complexity, it does not \reviewremove{meet our requirements} \reviewadd{scale well}.

\reviewadd{A protocol for shuffling based on secret sharing and homomorphic encryption is proposed in~\cite{KT06}. It is used as a subprotocol to anonymizes players' inputs prior to decentralized sorting and benchmarking. The ownership relation is concealed in a multi-round protocol where mix networks are used to ensure anonymity. Hence, it has drawbacks similar to those of mix networks. A constant-round benchmarking protocol for centralized scenarios based on homomorphic encryption is presented in~\cite{KER10}. Instead of sorting the full list of encrypted KPIs, it computes in a privacy-preserving fashion for each input the number of inputs that are smaller, such that no participant learns any KPI's rank. Even though this approach does not require shuffling to prevent leaking KPIs' ranks, it comes at the cost of quadratic computational and communication complexity, which implies poor scalability.}
\section{Secret Shuffling Protocol}\label{Protocol}

\begin{table*}[h!]
	\renewcommand{\arraystretch}{1.4}
	\centering
	\caption{Secret Shuffling Protocol with Step Labels and Computations}
	\label{tab:DesignAdaptedProtocolTable}
	\begin{tabular}[t]{p{.75cm}|p{12cm}}
		\hline \textbf{Step} & \textbf{Computation}\\\hline\hline
		\textbf{\ProtocolEncryptedInput{}} & \underline{$P_i \rightarrow P_S$:} $E_1(x_i)$\\
		
		\textbf{\ProtocolEncryptedPlayersRandoms{}} & $\!E_1(r_{1_i})$\\
		
		\hline
		
		\textbf{\ProtocolSequencePermutedRandomValues{}} & \underline{$P_S \rightarrow P_i$:} $R_1^\prime = (..., E_1(r_{1_i}^\prime) = E_1( r_{1_{\pi_1(i)}}), ...)$\\
		
		\textbf{\ProtocolSequenceEncryptedInputs{}} & $\!X^\prime = (..., E_1(x^\prime_i+r_{2_i}) = E_1(x_{\pi_2(i)}+r_{2_i}) = E_1(x_{\pi_2(i)}) \cdot E_1(r_{2_i}), ...)$\\
		
		\textbf{\ProtocolSequenceEncryptedInputsBlindings{}} & $\!R_2 = (..., E_2(r_{2_i}), ...)$\\
		
		\textbf{\ProtocolServiceProvidersRandom{}} &  $\!r_{{1_S}}$\\
		
		\textbf{\ProtocolSequenceHashes{}} & \underline{$P_i$:} $H = (..., h_j = h(r_{1_j}^\prime||r_{{1_S}}), ...)$\\
		
		\textbf{\ProtocolRandomPositionHash{}} & $\!\rho_i = position(H^\prime = sort(H), h_i)$\\
		
		\textbf{\ProtocolEncryptedRerandomisedBlindedInput{}} & \underline{$P_i \rightarrow P_S$:} $E_1^\prime(x^{\prime}_{\rho_i}+r_{2_{\rho_i}}) = E_1(x^{\prime}_{\rho_i}+r_{2_{\rho_i}}+0) = E_1(x^{\prime}_{\rho_i}+r_{2_{\rho_i}}) \cdot E_1(0)$\\
		
		\textbf{\ProtocolEncryptedBlindedBlinding{}} & $\!E_2(r_{2{\rho_i}}+r_{3_i}) = E_2(r_{2{\rho_i}}) \cdot E_2(r_{3_i})$\\
		
		\textbf{\ProtocolEncryptedBlindingBlinding{}} & $\!E_1(r_{3_i})$\\
		
		\textbf{\ProtocolSequenceRerandomisedPermutedInputs{}} & \underline{$P_S$:} $\mathcal{X} = (..., E_1^\prime(x^\prime_{\rho_i}) = E_1^\prime(x^{\prime}_{\rho_i}+r_{2_{\rho_i}}) \cdot E_1((-1) \cdot D_2(E_2(r_{2{\rho_i}}+r_{3_i}))) \cdot E_1(r_{3_i}), ...)$\\
		
		\hline
	\end{tabular}
\end{table*}

\subsection{Adversary Model}\label{ProtocolAdversaryModel}

We design our protocol to be secure against any semi-honest adversary $\mathcal{A}$~\cite{Lin17} that corrupts either an arbitrary number of players or the service provider. That is, we exclude collusion between any player and the service provider, like the related work. Our shuffling protocol ensures input privacy. Hence, $\mathcal{A}$ does not learn anything about non-corrupted players' secret inputs. Most importantly, we ensure that no such $\mathcal{A}$ is able to map non-corrupted players' inputs to their equivalents in the shuffled sequence generated by the shuffling protocol. In summary, no adversary corrupting either any subset of the players or the service provider can determine the ownership relation between non-corrupted players or their secret inputs.

\subsection{Prerequisites}\label{ProtocolPrerequisites}

In the description of our protocol SHUFFLE, we use the indices $1 \leq i,j \leq n$ for players $P_i$ and $P_j$, respectively, as well as their inputs $x_i$ ($x_j$), random values $r_i$ ($r_j$), etc. We denote concatenation by ``$||$''.

We assume two instances $\mathcal{CS}_1$ and $\mathcal{CS}_2$ of the DamgÃ¥rd-Jurik cryptosystem~\cite{DAJ01}, like Paillier's~\cite{PAI99}.
The public keys $pk_1$ and $pk_2$ are known to the service provider $P_S$ and the players $P_i$. The secret key $sk_1$ is known only to the players and could be generated and distributed via Diffie-Hellman key exchange~\cite{DH76}. The secret key $sk_2$ is only known to $P_S$.
We require the plaintext space $\mathcal{M}_2$ of $\mathcal{CS}_2$ to be a subset of the plaintext space $\mathcal{M}_1$ of $\mathcal{CS}_1$, i.e.,
\begin{equation}\label{Eq:Ciphertextspace1EqualsPlaintextspace2}
	\mathcal{M}_2 \subseteq \mathcal{M}_1.
\end{equation}
This ensures that any message that can be encrypted with $pk_2$ can also be encrypted with $pk_1$.

We require two random permutations $\pi_1$ and $\pi_2$, a cryptographic hash function $h(\cdot)$, and two functions $sort(S)$ and $position(S,s_i)$. The permutations $\pi_1$ and $\pi_2$ are both chosen by and only known to $P_S$. The hashes of $h(\cdot)$ are assumed to be uniformly distributed among the domain $dom(h(\cdot))$. Given a sequence $S = (s_1,...,s_n)$, $sort(S)$ outputs a sequence $S^\prime$ that contains $s_1,...,s_n$ in ascending order, as in~\eqref{Eq:SortedSequenceElementsS} and~\eqref{Eq:SortedSequenceElementsSPrime}.
\begin{equation}\label{Eq:SortedSequenceElementsS}
	S^\prime = (...,s_i,... | s_i \in S)
\end{equation}
\begin{equation}\label{Eq:SortedSequenceElementsSPrime}
	S^\prime = (s^\prime_1 \leq s^\prime_2 \leq \cdots \leq s^\prime_n)
\end{equation}
The function $position(S,s_i)$ outputs the position of $s_i$ in $S$.

Moreover, we assume pairwise secure, i.e., secret and authentic, channels between each player and the service provider, for instance established via Transport Layer Security (TLS).

\subsection{Protocol Specification}\label{ProtocolSpecification}

According to Definition~\ref{Def:SecureShuffle}, for a protocol to secretly shuffle a sequence, it has to permute the order of the entries by a random permutation $\pi$. Furthermore, it has to change the ciphertexts of the secret inputs such that $\pi$ cannot be reconstructed. To achieve this, each player performs two main tasks: randomly but uniquely selecting some player's encrypted input and rerandomizing (see Equation~\eqref{Eq:Rerandomisation}) this input. The former is based on a novel approach to random index distribution. For this random index distribution, each player provides a random input, which is concatenated with a random value given by the service provider. The resulting concatenations are then hashed and the hashes are sorted. The position of the hash in the sorted list of hashes corresponding to a player's random input is its random index. Our protocol runs in two communication rounds. The first round is used for collecting the players' inputs and the second round conducts the actual shuffling. It is depicted in Table~\ref{tab:DesignAdaptedProtocolTable} and explained in Section~\ref{ProtocolExplanation}.

\subsection{Protocol Explanation}\label{ProtocolExplanation}

In step~\ProtocolEncryptedInput{}, each player sends its private input $x_i$ that is supposed to be shuffled, encrypted under $\mathcal{CS}_1$. Then, in step~\ProtocolEncryptedPlayersRandoms{}, each player chooses a (presumably unique) random value $r_{1_i} \in \mathcal{M}_1$ and sends it to $P_S$, encrypted under $\mathcal{CS}_1$. This random value will be used for random index distribution. Hence, the service provider receives two list of $n$ ciphertexts.%

The service provider then, in step~\ProtocolSequencePermutedRandomValues{}, sends the full list $R^\prime_1$ of encrypted random values $E_1(r_{1_i})$ to the players. Permutation via $\pi_1$ prevents the players from learning which $r_{1_i}$ was provided by which player. Similarly, in step~\ProtocolSequenceEncryptedInputs{}, it sends the full list of encrypted input values $E_1(x_i)$, permuted via $\pi_2$. To prevent the players from learning the secret inputs, each plaintext $x_{\pi_2(i)}$ is blinded by a value $r_{2_i} \in \mathcal{M}_1$, chosen individually and at random for each $i$ by the service provider. %
The full list of random values $r_{2_i}$, encrypted under $\mathcal{CS}_2$, is sent to the players in step~\ProtocolSequenceEncryptedInputsBlindings{}. Then, $P_S$ chooses one long random value $r_{{1_S}}$, e.g., $r_{{1_S}} \in \mathcal{M}_1$, and sends it to the players in step~\ProtocolServiceProvidersRandom{}.
Hence, the players receive the same three lists of $n$ ciphertexts and the same random value.

In step~\ProtocolSequenceHashes{}, each $P_i$ decrypts the ciphertexts $E_1(r_{1_j}^\prime) \in R^\prime_1$, $1 \leq j \leq n$. If the values $r_{1_j}^\prime$ are not unique, the players abort the protocol. Otherwise, each player concatenates each resulting plaintext $r_{1_j}^\prime$ with the random value $r_{{1_S}}$ of $P_S$ and computes the $n$ hashes $h_j$. Using $r_{{1_S}}$ as a seed of the hash function prevents any player $P_i$ from selecting a specific $r_{1_i}$ in step~\ProtocolEncryptedPlayersRandoms{} to obtain a desired hash $h_i$, which would eventually affect the (random) index distribution.
In step~\ProtocolRandomPositionHash{}, each $P_i$ sorts the list of hashes. For the hash $h_i=h(r_i||r_{{P_S}})$ corresponding to player $P_i$'s random value $r_{1_i}$, the position $\rho_i$ in the sorted list of hashes is the random index of $P_i$. Hence, each player computes an individual, random index $\rho_i$ that is unknown to $P_S$ and not related to the rank of its input $x_i$.

Given $\rho_i$, each player sends the ciphertext $E_1(x^{\prime}_{\rho_i}+r_{2_{\rho_i}})$ to $P_S$ in step~\ProtocolEncryptedRerandomisedBlindedInput{}. To prevent the service provider from learning $\rho_i$, this ciphertext is rerandomized. Additionally, in step~\ProtocolEncryptedBlindedBlinding{}, the encrypted random value of index $\rho_i$ in $R_2$, i.e., $E_2(r_{2_i})$, is sent to $P_S$. The underlying plaintext $r_{2_i}$ is blinded by a random value $r_{3_i}$. This random value $r_{3_i}$, encrypted under $\mathcal{CS}_1$, is then sent to $P_S$ in step~\ProtocolEncryptedBlindingBlinding{}. Hence, the service provider receives three ciphertexts from each player.

In step~\ProtocolSequenceRerandomisedPermutedInputs{}, the service provider decrypts the ciphertexts $E_2(r_{2{\rho_i}} + r_{3_i})$ received in step~\ProtocolEncryptedBlindedBlinding{}, multiplies the resulting plaintexts with $-1$, and encrypts the products under cryptosystem $\mathcal{CS}_1$. The resulting ciphertexts are multiplied with the ciphertexts $E_1(x^{\prime}_{\rho_i}+r_{2_{\rho_i}})$ of step~\ProtocolEncryptedRerandomisedBlindedInput{} and $E_1(r_{3_i})$ of step~\ProtocolEncryptedBlindingBlinding{}. Consequently, the random values $r_{2_i}$ and $r_{3_i}$ are eliminated, resulting in rerandomized ciphertexts $\chi_i = E_1(x^{\prime}_{\rho_i})$.

The order of the rerandomized ciphertexts $\chi_i$ of the input values $x_i$ is determined by the input order of the values in steps~\ProtocolEncryptedRerandomisedBlindedInput{} to~\ProtocolEncryptedBlindingBlinding{} as received via network. Every $P_i$ sends some $P_j$'s rerandomized, encrypted input, chosen based on its random index. The service provider cannot map between the original input order and the order of $\mathcal{X}$. %
Therefore, $P_S$'s output is a shuffled list. The players do not get an output.
\section{Proof of Input Privacy}\label{PrivacyProof}

We denote privacy by a tuple $(a,b)$, stating that $a$ players or (exclusively) $b$ service providers can be corrupted without input privacy being at risk. We prove that the players' inputs in the protocol SHUFFLE are $(t,1)$-private against semi-honest adversaries $\mathcal{A}$. This is formalized as follows.

\begin{theorem}[Input Privacy]\label{Thrm:InputPrivacy}
	The protocol SHUFFLE $(t,1)$-privately computes the shuffled sequence $\mathcal{X} = (E_{\reviewadd{1}}^\prime(x_{\pi(1)}),...,E_{\reviewadd{1}}^\prime(x_{\pi(n)}))$ from the input sequence $(x_1,...,x_n)$ in the semi-honest model as long as there is no collusion between any player and the service provider.
\end{theorem}

First, we define the view of a participant as follows~\cite{Gol02}.

\begin{definition}[View]\label{Def:View}
	A participant $P_i$'s view $V_i(x_1,...,x_n) = \{x_i,r_i,m_{i_1},...,m_{i_\phi}\}$ in the execution of a protocol $\Pi$ on inputs $(x_1,...,x_n)$ contains $P_i$'s input $x_i$, $P_i$'s internal random tape $r_i$, and any message $m_{i_k}$ that $P_i$ receives during execution of $\Pi$.
\end{definition}
  
For a secure computation protocol to be secure in the semi-honest model, it is sufficient to prove that anything an adversary $\mathcal{A}$ can learn during protocol execution can as well be learned given only the inputs and outputs of the protocol~\cite{Lin17}.
That is, it is sufficient to show that the view of $\mathcal{A}$ can be generated by some polynomial-time algorithm $\mathcal{S}$, called simulator, entirely based on the inputs and outputs of the $t$ corrupted players or the exclusively corrupted service provider. This can be formalized as follows~\cite{Gol02}.

\begin{definition}[Functionality, Simulator, Privacy]\label{def:PrivacySimulator}
	Let $f(x_1,...,x_n):(\{0,1\}^\ast)^n \rightarrowtail (\{0,1\}^\ast)^n$ be the shuffling functionality. For $I=\{i_1,...,i_t\} \subset \{1,...,n\}$ let $V_I(x_1,...,x_n)=(I,V_{i_1}(x_1,...,x_n),...,V_{i_t}(x_1,...,x_n))$. The protocol SHUFFLE $(t,1)$-privately computes $f(x_1,...,x_n)$ if there exists a polynomial-time simulator $\mathcal{S}$ that, given the corrupted participants' inputs and output, generates an output that is computationally indistinguishable from $V_I(x_1,...,x_n)$ for any $I$, i.e.\reviewadd{,} $\mathcal{S}(I,(x_{i_1},...,x_{i_t}),f(x_1,...,x_n)) \stackrel{c}{\equiv} V_I(x_1,...,x_n)$.
\end{definition}

\paragraph{Proof Outline}\label{PrivacyProofOutline}

Our protocol has two different kinds of participants: $n$ players with an input but no output and one service provider with no input but an output. Hence, we have two different views that need to be simulated by two different simulators. They simulate the protocol inputs by taking the inputs from the real protocol execution and simulate the coin tosses by using the same pseudo-random generator (PRG) as in the real protocol execution. This results in a simplified view that only contains the messages $m_{i_k}$, which the corrupted participants receive.
We prove our protocol to $1$-privately compute the shuffling functionality $f$ in case an adversary $\mathcal{A}$ corrupts only the service provider. Additionally, we prove that the protocol SHUFFLE $t$-privately computes $f$ in case an adversary $\mathcal{A}$ corrupts $t$ players but not the service provider. This leads to the two Lemmas~\ref{Lem:PrivacyPlayers} and~\ref{Lem:PrivacyServiceProvider}.

\begin{lemma}[Input Privacy -- Players]\label{Lem:PrivacyPlayers}
	The protocol SHUFFLE $t$-privately computes the shuffled sequence $\mathcal{X} = (E_{\reviewadd{1}}^\prime(x_{\pi(1)}),...,E_{\reviewadd{1}}^\prime(x_{\pi(n)}))$ from the input sequence $(x_1,...,x_n)$ for semi-honest adversaries that corrupt $t$ players but not the service provider.
\end{lemma}

\begin{proof} The proof of Lemma~\ref{Lem:PrivacyPlayers} gives the players' view and simulator. Then, the computational indistinguishability of the view and the simulator's output is shown.
	
	Each player $P_i$ provides as input a secret value $x_i$ and does not get an output. The players have the secret decryption key $sk_1$ and can decrypt any ciphertext $c_j = E_1(x_j)$. An arrow ``$\rightarrow$'' shows the corresponding plaintexts that the players can compute given $sk_1$.
	Each $P_i$ receives the following messages in the respective protocol steps.
	\begin{enumerate}
		\renewcommand{\labelenumi}{\textbf{\theenumi}}
		\renewcommand{\theenumi}{2.\arabic{enumi}}
		
		\item $E_1(r_{\reviewadd{1_1}}^\prime),...,E_1(r_{\reviewadd{1_n}}^\prime) \rightarrow r_{\reviewadd{1_1}}^\prime,...,r_{\reviewadd{1_n}}^\prime$
		
		\item $E_1(x_1^\prime+r_{2_1}),...,E_1(x_n^\prime+r_{2_n}) \rightarrow x_1^\prime+r_{2_1},...,x_n^\prime+r_{2_n}$
		
		\item $E_2(r_{2_1}),...,E_2(r_{2_n})$
		
		\item $r_{1_S}$
	\end{enumerate}
	
	If a message can be decrypted, the players' simulator $\mathcal{S}_{P_i}$ simulates the underlying plaintext. 
	Encryption can be regarded as a deterministic mapping of probability distributions~\cite{KER10}. Hence, if the computationally indistinguishable simulation of a plaintext is possible, so is the computationally indistinguishable simulation of the corresponding ciphertext.
	Given that $dom(\cdot)$ denotes the domain of a function, the players' simulator $\mathcal{S}_{P_i}$ generates the following simulated messages.
	\begin{enumerate}
		\renewcommand{\labelenumi}{\textbf{\theenumi}}
		\renewcommand{\theenumi}{2.\arabic{enumi}}
		
		\item $n$ random values $r_{\ProtocolSequencePermutedRandomValues{}_1},...,r_{\ProtocolSequencePermutedRandomValues{}_n}$, uniformly chosen from $dom(D_1(\cdot)) = \mathcal{M}_1$
		
		\item $n$ random values $r_{\ProtocolSequenceEncryptedInputs{}_1},...,r_{\ProtocolSequenceEncryptedInputs{}_n}$, uniformly chosen from $dom(D_1(\cdot)) = \mathcal{M}_1$
		
		\item $n$ random values $r_{\ProtocolSequenceEncryptedInputsBlindings{}_1},...,r_{\ProtocolSequenceEncryptedInputsBlindings{}_n}$, uniformly chosen from $dom(E_2(\cdot)) = \mathcal{C}_2$
		
		\item A random value $r_{\ProtocolServiceProvidersRandom{}}$, uniformly chosen from $dom(D_1(\cdot)) = \mathcal{M}_1$
	\end{enumerate}
	
	We show that the simulator's output and the players' view are computationally indistinguishable. To prove computational indistinguishability of a real and a simulated message, one needs to show that the probability distribution of the real message is known to the simulator and that the function generating the simulated output is identically distributed~\cite{KER10}.
	\\
	In step~\ProtocolSequencePermutedRandomValues{}, the values sent are the random values that were uniformly chosen by the players in step~\ProtocolEncryptedPlayersRandoms{} from $dom(D_1(\cdot))=\mathcal{M}_1$ and encrypted with $pk_1$. The players can decrypt these messages, resulting in the original random values. The values generated by the simulator, which are chosen uniformly at random from \reviewremove{$dom(\mathcal{M}_1) = dom(D_1(\cdot))$} \reviewadd{$\mathcal{M}_1 = dom(D_1(\cdot))$}, are identically distributed.
	\\
	In step~\ProtocolSequenceEncryptedInputs{}, the values sent are sums of the players' secret inputs and random values that were uniformly chosen by the service provider from $\mathcal{M}_1 = dom(D_1(\cdot))$ and encrypted with $pk_1$. The players can decrypt these messages, resulting in the original sums with one summand being a uniformly chosen random number. Hence, they are identically distributed to the simulator's output, which are values chosen uniformly at random from \reviewremove{$dom(\mathcal{M}_1) = dom(D_1(\cdot))$} \reviewadd{$\mathcal{M}_1 = dom(D_1(\cdot))$}.
	\\
	The message of step~\ProtocolSequenceEncryptedInputsBlindings{} cannot be decrypted by $P_i$. The DamgÃ¥rdâ€“Jurik cryptosystem ensures semantic security~\cite{DAJ01}. Hence, ciphertexts are computationally indistinguishable from values that are chosen uniformly at random from the ciphertext space. The real messages in $dom(E_2(\cdot)) = \mathcal{C}_2$ are computationally indistinguishable from the simulator's outputs, chosen uniformly at random from $dom(E_2(\cdot)) = \mathcal{C}_2$.
	\\
	In step~\ProtocolServiceProvidersRandom{}, the real message is a value that was chosen uniformly at random from $dom(D_1(\cdot)) = \mathcal{M}_1$. The simulator's output, chosen uniformly at random from $dom(D_1(\cdot)) = \mathcal{M}_1$, and the real message are identically distributed.
	\\
	Given these comparisons, the described simulator for the players generates an output that is computationally indistinguishable from real views.
	\\
	This completes the proof of Lemma~\ref{Lem:PrivacyPlayers}.%
\end{proof}

\begin{lemma}[Input Privacy -- Service Provider]\label{Lem:PrivacyServiceProvider}
	The protocol SHUFFLE $1$-privately computes the shuffled sequence $\mathcal{X} = (E_{\reviewadd{1}}^\prime(x_{\pi(1)}),...,E_{\reviewadd{1}}^\prime(x_{\pi(n)}))$ from the input sequence $(x_1,...,x_n)$ for semi-honest adversaries that corrupt the service provider.
\end{lemma}

\begin{proof} The proof of Lemma~\ref{Lem:PrivacyServiceProvider} gives the service provider's view and simulator. Then, the computational indistinguishability of the view and the simulator's output is shown.

	The service provider $P_S$ does not have an input. Its output \reviewremove{is} \reviewadd{are} the permuted, rerandomized encryptions of the players' inputs. It knows the secret decryption key $sk_2$ and can decrypt any $c_j = E_2(x_j)$. An arrow ``$\rightarrow$'' shows the plaintexts that it can compute given $sk_2$.
	It receives the following messages.
	\begin{enumerate}
		\renewcommand{\labelenumi}{\textbf{\theenumi}}
		\renewcommand{\theenumi}{1.\arabic{enumi}}
		
		\item $E_1(x_i)$
	
		\item $E_1(r_{1_i})$
		
		\renewcommand{\theenumi}{2.\arabic{enumi}}
		\setcounter{enumi}{6}
		
		\item $E_1^\prime(x_{\rho_i}^\prime+r_{2_{\rho_i}})$
		
		\item $E_2(r_{2_{\rho_i}}+r_{3_i}) \rightarrow r_{2_{\rho_i}}+r_{3_i}$
		
		\item $E_1(r_{3_i})$
	\end{enumerate}
	
	If a message can be decrypted, the service provider's simulator $\mathcal{S}_{P_S}$ simulates the underlying plaintext. %
	It generates the following simulated messages.
	\begin{enumerate}
		\renewcommand{\labelenumi}{\textbf{\theenumi}}
		\renewcommand{\theenumi}{1.\arabic{enumi}}
		
		\item A random value $r_{\ProtocolEncryptedInput{}}$, uniformly chosen from $dom(E_1(\cdot)) = \mathcal{C}_1$
		
		\item A random value $r_{\ProtocolEncryptedPlayersRandoms{}}$, uniformly chosen from $dom(E_1(\cdot)) = \mathcal{C}_1$
		
		\renewcommand{\theenumi}{2.\arabic{enumi}}
		\setcounter{enumi}{6}
		
		\item A random value $r_{\ProtocolEncryptedRerandomisedBlindedInput{}}$, uniformly chosen from $dom(E_1(\cdot)) = \mathcal{C}_1$
		
		\item A random value $r_{\ProtocolEncryptedBlindedBlinding{}}$, uniformly chosen from $dom(D_2(\cdot)) = \mathcal{M}_2$
		
		\item A random value $r_{\ProtocolEncryptedBlindingBlinding{}}$, uniformly chosen from $dom(E_1(\cdot)) = \mathcal{C}_1$
	\end{enumerate}
	
	We show that the simulator's output and the service provider's views are computationally indistinguishable. %
	\\
	The messages of steps~\ProtocolEncryptedInput{},~\ProtocolEncryptedPlayersRandoms{},~\ProtocolEncryptedRerandomisedBlindedInput{}, and~\ProtocolEncryptedBlindingBlinding{} are ciphertexts in $dom(E_1(\cdot)) = \mathcal{C}_1$, which cannot be decrypted by $P_S$. Based on the semantic security of the DamgÃ¥rdâ€“Jurik cryptosystem, the real messages are computationally indistinguishable from the simulator's output, which are random values uniformly chosen from $dom(E_1(\cdot)) = \mathcal{C}_1$. %
	\\
	In step~\ProtocolEncryptedBlindedBlinding{}, the value sent is the sum of a player's secret input and a random value that was chosen by $P_i$ uniformly from $\mathcal{M}_2 = dom(D_2(\cdot))$ and encrypted with $pk_2$. The service provider can decrypt this message, resulting in the original sum with one summand being a uniformly chosen random number. Therefore, it is identically distributed to the simulator's output, which is a value chosen uniformly at random from \reviewremove{$dom(\mathcal{M}_2) = dom(D_2(\cdot))$} \reviewadd{$\mathcal{M}_2 = dom(D_2(\cdot))$}.
	\\
	Given these comparisons, the described simulator for the service provider generates an output that is computationally indistinguishable from real views. \\
	This completes the proof of Lemma~\ref{Lem:PrivacyServiceProvider}.%
\end{proof}
\section{Proof of Correctness}\label{CorrectnessProof}

Correctness of our protocol is shown by proving Theorem~\ref{Thrm:ProtocolIsSecureShuffle}.

\begin{theorem}[Correctness]\label{Thrm:ProtocolIsSecureShuffle}
	The protocol SHUFFLE conducts a secret shuffle of the $n$ players' encrypted inputs. That is, for every sequence of ciphertexts $X = (...,E_{\reviewadd{1}}(x_i),...)$ with $1 \leq i \leq n$, the protocol SHUFFLE yields as output a sequence $\mathcal{X} = (...,E_{\reviewadd{1}}^\prime(x_{\pi(i)}),...)$ such that the ciphertexts $E_{\reviewadd{1}}^\prime(x_i) \neq E_{\reviewadd{1}}(x_i)$ have the same plaintexts $x_i$, but their order in $\mathcal{X}$ is randomly permuted by a permutation $\pi$. The permutation $\pi$ is not known to any participant as long as there is no collusion between any player and the service provider.
\end{theorem}

To improve readability of the proof of correctness, we split Theorem~\ref{Thrm:ProtocolIsSecureShuffle} into the three Lemmas according to the properties of a secret shuffle as given in Definition~\ref{Def:SecureShuffle}, which we will prove separately. First, we will prove that the protocol SHUFFLE outputs a randomly permuted sequence (Lemma~\ref{Lem:PermutedCiphertexts}). Then, we show that the ciphertexts in the output sequence are different from those of the input sequence but encrypt the same plaintexts (Lemma~\ref{Lem:SamePlaintexts}). We complete the proof of correctness by proving that no participant learns the overall, random permutation (Lemma~\ref{Lem:PermutationPrivacy}).

\subsection{Randomly Permuted Ciphertexts}

\begin{lemma}[Randomly Permuted Ciphertexts]\label{Lem:PermutedCiphertexts}
	The ciphertexts in sequence $\mathcal{X} = (...,E_{\reviewadd{1}}^\prime(x_{\pi(i)}),...)$ output by the protocol SHUFFLE are permuted compared to the ciphertexts in the input sequence $X = (...,E_{\reviewadd{1}}(x_i),...)$ with a random permutation $\pi$.
\end{lemma}

We prove Lemma~\ref{Lem:PermutedCiphertexts} by showing that the encrypted inputs sent in step~\ProtocolEncryptedRerandomisedBlindedInput{} are selected based on unique, random indices.

\begin{proof}\label{Prf:PermutedCiphertexts}
	\reviewadd{The order of sequence $X = (...,E_1(x_i),...)$ is determined by the order in which the service provider $P_S$ receives these inputs (step~\ProtocolEncryptedInput{}). The order of sequence $X^\prime = (...,E_1(x_{\pi_2(i)}+r_{2_i}),...)$ is determined by the random permutation $\pi_2$ (step~\ProtocolSequenceEncryptedInputs{}).}
	The order of sequence $R = (...,E_1(r_{1_i}),...)$ is determined by the order in which the service provider $P_S$ receives these inputs (step~\ProtocolEncryptedPlayersRandoms{}). The order of sequence $R^\prime = (...,E_1(r_{1_{\pi_1(i)}}),...)$ is determined by the random permutation $\pi_1$ (step~\ProtocolSequencePermutedRandomValues{}).
	\\
	The values $r_{1_i}$ and $r_{P_S}$ (step~\ProtocolServiceProvidersRandom{}) are chosen uniformly at random and known to each player. Starting from step~\ProtocolSequenceHashes{}, the random values $r_{1_i}$ are guaranteed to be distinct.
	Therefore, and as $h(\cdot)$ is a cryptographic hash function, ensuring collision resistance (see Section~\ref{ProtocolPrerequisites}), it follows that the $n$ hashes
	\begin{equation}\label{Eq:HashRandomValues}
		\begin{aligned}
			h_i &= h(r_{1_i}||r_{{1_S}})\\
			&= h(D_1(E_1(r_{1_i}))||r_{{1_S}})
		\end{aligned}
	\end{equation}
	in the sequence $H = (...,h_i,...)$ are distinct, except with negligible probability. Since the hashes of $h(\cdot)$ are required to be uniformly distributed among the domain $dom(h(\cdot))$ \reviewadd{(see Section~\ref{ProtocolPrerequisites})}, the values in $H$ are uniformly distributed among $dom(h(\cdot))$ too. Since every player $P_i$ knows its random value $r_{1_i}$ of step~\ProtocolEncryptedPlayersRandoms{} and $r_{1_S}$, \reviewadd{each} $P_i$ also knows its unique corresponding hash $h_i \in H$.
	\\
	The function $sort(\cdot)$ sorts a sequence in ascending order. Hence, the sequence $H^\prime = sort(H)$, computed in step~\ProtocolRandomPositionHash{} contains the same values as the sequence $H$ but sorted in ascending order. Therefore, each player $P_i$ knows the position of its hash in the sorted list of hashes $H^\prime$. This position is extracted with the function $position(H^\prime, h_i)$, which thus provides the correct index $\rho_i$ of $P_i$'s hash $h_i$ in $H^\prime$.
	\\
	Assume $\rho_i$ can be distinguished from a random element in $\{1,...,n\}$. This implies that the permutation applied by the function $sort(\cdot)$ can be distinguished from a random permutation. Sorting a sequence of distinct random values produces a random permutation over the random input values~\cite{CLRS09}. According to the above assumption, if the result of $sort(\cdot)$ is distinguishable from a random permutation, then the result of $h(\cdot)$ is distinguishable from random values as well. However, as a cryptographically secure hash function is a random oracle, this contradicts the assumption of a random oracle. According to this contradiction, the random indices $\rho_i$ are computationally indistinguishable from values chosen uniformly at random from $\{1,...,n\}$. Since the hashes in $H$, and therefore also the hashes in $H^\prime$, are distinct, the $n$ values $\rho_i \in \{1,...,n\}$ are also distinct.
	\\
	Consequently, the ciphertext $E_1(x^{\prime}_{\rho_i}+r_{2_i})$ selected from $X^\prime$ by player $P_i$ in step~\ProtocolEncryptedRerandomisedBlindedInput{} based on $\rho_i$ is randomly and exclusively selected. %
	\\
	The order of sequence $\mathcal{X} = (...,E_1^\prime(x^\prime_{\rho_i}),...)$ (step~\ProtocolSequenceRerandomisedPermutedInputs{}) is determined by the order in which the service provider receives the ciphertexts of steps~\ProtocolEncryptedRerandomisedBlindedInput{} to~\ProtocolEncryptedBlindingBlinding{} from the players. The ciphertext sent by $P_i$ encrypts some $P_j$'s secret input chosen \reviewadd{from a randomly permuted sequence} based on its uniformly distributed, unique index $\rho_i$. Hence, the order of encrypted %
	inputs in $\mathcal{X}$ is randomly permuted.
	\\
	This completes the proof of Lemma~\ref{Lem:PermutedCiphertexts}.%
\end{proof}

\subsection{Distinct Ciphertexts}

\begin{lemma}[Distinct Ciphertexts]\label{Lem:SamePlaintexts}
	The ciphertexts in sequence $\mathcal{X} = (...,E_{\reviewadd{1}}^\prime(x_{\pi(i)}),...)$ output by the protocol SHUFFLE encrypt the same plaintexts $x_i$ as the ciphertexts in the input sequence $X = (...,E_{\reviewadd{1}}(x_i),...)$. The ciphertexts $E^\prime_1(x_i) \in \mathcal{X}$ and $E_1(x_i) \in X$ encrypting the same plaintext $x_i$ are distinct, i.e.\reviewadd{,} $E_{\reviewadd{1}}^\prime(x_1) \neq E_{\reviewadd{1}}(x_1),...,E_{\reviewadd{1}}^\prime(x_n) \neq E_{\reviewadd{1}}(x_n)$, and cannot be mapped to each other by the service provider.
\end{lemma}

We prove Lemma~\ref{Lem:SamePlaintexts} by showing that the operations performed on the uniquely and randomly selected ciphertexts change the ciphertexts without affecting their plaintexts.

\begin{proof}\label{Prf:SamePlaintexts}
	In step~\ProtocolSequenceEncryptedInputs{}, sequence $X = (...,E_1(x_i),...)$ of the input ciphertexts is randomly permuted with a permutation $\pi_2$ and the underlying plaintexts are blinded with a random value $r_{2_i}$, resulting in the sequence
	\begin{equation}\label{Eq:SequenceXPrime}
		\begin{aligned}
			X^\prime &= \left(..., E_1\left(x_{\pi_2(i)}\right) \cdot E_1\left(r_{2_i}\right), ...\right)\\
			&= \left(..., E_1\left(x_{\pi_2(i)}+r_{2_i}\right), ...\right)\\
			&= \left(..., E_1\left(x^\prime_i+r_{2_i}\right), ...\right).
		\end{aligned}
	\end{equation}
	Permutations only affect the order of a sequence's elements but not the elements themselves (see Section~\ref{Preliminaries}). Hence, sequence $X^\prime$ contains encryptions of the original $n$ \reviewremove{ciphertexts $E_(x_i)$} \reviewadd{plaintexts $x_i$}, blinded with $n$ random values $r_{2_i}$.
	\\
	In step~\ProtocolEncryptedRerandomisedBlindedInput{}, for each player $P_i$ with individual index $\rho_i$ (see Proof of Lemma~\ref{Lem:PermutedCiphertexts}), the ciphertext $E_1(x^\prime_i+r_{2_i}) \in X^\prime$ is rerandomized by multiplication with $E_1(0)$. The result is the different ciphertext $E_1^\prime(x^{\prime}_{\rho_i}+r_{2_{\rho_i}}) \neq E_1\left(x^\prime_i+r_{2_i}\right)$ of the same plaintext (see~\eqref{Eq:CorrectnessProofRerandomisation}).
	\begin{equation}\label{Eq:CorrectnessProofRerandomisation}
	\begin{aligned}
		E_1^\prime\left(x^{\prime}_{\rho_i}+r_{2_{\rho_i}}\right) &= E_1\left(x^{\prime}_{\rho_i}+r_{2_{\rho_i}}\right) \cdot E_1\left(0\right)\\
		&= E_1\left(x^{\prime}_{\rho_i}+r_{2_{\rho_i}}+0\right)
	\end{aligned}
	\end{equation}
	As the service provider knows the secret decryption key $sk_2$, in step~\ProtocolSequenceRerandomisedPermutedInputs{}, it can decrypt the ciphertext $E_2(r_{2{\rho_i}}+r_{3_i})$, negate it via multiplication by $-1$, and encrypt \reviewadd{it} with $pk_1$. Multiplying the resulting ciphertext $E_1(-r_{2{\rho_i}}-r_{3_i})$ with the ciphertexts of steps~\ProtocolEncryptedRerandomisedBlindedInput{} and~\ProtocolEncryptedBlindingBlinding{} yields the rerandomized, unblinded ciphertext $E_1^\prime(x^\prime_{\rho_i})$ of an input value $x^\prime_{\rho_i}$ as follows.
	\begin{equation}
	\begin{aligned}\label{Eq:CorrectnessProofRerandomisationResult}
		E_1^\prime\left(x^\prime_{\rho_i}\right) &= E_1^\prime\left(x^{\prime}_{\rho_i}+r_{2_{\rho_i}}+0\right)\\
			& \hspace{.5cm} \cdot E_1\left(\left(-1\right) \cdot D_2\left(E_2\left(r_{2{\rho_i}}+r_{3_i}\right)\right)\right)\\
			& \hspace{.5cm} \cdot E_1\left(r_{3_i}\right)\\
		&= E_1^\prime\left(x^{\prime}_{\rho_i}+r_{2_{\rho_i}}+0\right)\\
			& \hspace{.5cm} \cdot E_1\left(\left(-1\right) \cdot \left(r_{2{\rho_i}}+r_{3_i}\right)\right) \cdot E_1\left(r_{3_i}\right)\\
		&= E_1^\prime\left(x^{\prime}_{\rho_i}+r_{2_{\rho_i}}+0\right)\\
			& \hspace{.5cm} \cdot E_1\left(-r_{2{\rho_i}}-r_{3_i}\right) \cdot E_1\left(r_{3_i}\right)\\
		&= E_1^\prime\left(x^{\prime}_{\rho_i}+r_{2_{\rho_i}}+0-r_{2{\rho_i}}-r_{3_i}+r_{3_i}\right)\\
		&= E_1^\prime\left(x^{\prime}_{\rho_i}+0\right)
	\end{aligned}
	\end{equation}
	These rerandomized ciphertexts $E_1^\prime(x^\prime_{\rho_i})$ form sequence $\mathcal{X} = (...,E_1^\prime(x^\prime_{\rho_i}),...)$. From~\eqref{Eq:SequenceXPrime},~\eqref{Eq:CorrectnessProofRerandomisation}, and~\eqref{Eq:CorrectnessProofRerandomisationResult}, it follows that these ciphertexts encrypt the same plaintexts as the ciphertexts $E_1(x^\prime_{\rho_i}) \in X^\prime$ and therefore also the same plaintexts as the ciphertexts $E_1(x_i) \in X$.
	\\
	From~\eqref{Eq:CorrectnessProofRerandomisation}, it follows that the rerandomized ciphertexts $E_{\reviewadd{1}}^\prime(x_i) \in \mathcal{X}$ and the non-rerandomized ciphertexts $E_{\reviewadd{1}}(x_i) \in X$ that encrypt the same secret input $x_i$ are different from each other, i.e.\reviewadd{,} $E_{\reviewadd{1}}^\prime(x_1) \neq E_{\reviewadd{1}}(x_1),...,E_{\reviewadd{1}}^\prime(x_n) \neq E_{\reviewadd{1}}(x_n)$. As the service provider does not learn the probabilistic encryptions $E_{\reviewadd{1}}(0)$ used in step~\ProtocolEncryptedRerandomisedBlindedInput{} for rerandomization, it cannot invert the rerandomization of the ciphertexts $E_1(x^{\prime}_{\rho_i}+r_{2_{\rho_i}})$. Consequently, the service provider cannot map the rerandomized ciphertexts in $\mathcal{X}$ to the original ciphertexts in $X$.
	\\
	This completes the proof of Lemma~\ref{Lem:SamePlaintexts}.%
\end{proof}

\subsection{Secret Permutation}

The overall permutation $\pi$ applied during the protocol SHUFFLE consists of the following two composed, independent permutations.
\begin{itemize}
	\item The permutation $\pi_2$, applied by the service provider in step~\ProtocolSequenceEncryptedInputs{} to permute $X^\prime$.
	\item The permutation which the players implicitly apply to the output sequence $\mathcal{X}$ by selecting a ciphertext $E_1(x^{\prime}_{\rho_i}+r_{2_{\rho_i}})$ based on their random indices. We denote it by $\pi_3$.
\end{itemize}

That leads to the overall permutation $\pi(i) = \pi_3(\pi_2(i))$. 
As there are two different kinds of participants, we prove secrecy of $\pi$ separately for the players (Lemma~\ref{Lem:PermutationPrivacyPlayer}) and for the service provider (Lemma~\ref{Lem:PermutationPrivacyServiceProvider}).

\begin{lemma}[Secret Permutation -- Players]\label{Lem:PermutationPrivacyPlayer}
	The protocol SHUFFLE computes the shuffled sequence $\mathcal{X} = (...,E_{\reviewadd{1}}^\prime(x_{\pi(i)}),...)$ from the input sequence $X = (...,E_{\reviewadd{1}}(x_i),...)$ based on a random permutation $\pi$ such that no player learns the permutation $\pi$ as long as there is no collusion between any player and the service provider.
\end{lemma}

\begin{proof}
	Permutation $\pi_2$ is chosen at random by the service provider in step~\ProtocolSequenceEncryptedInputs{}. Collusion between players and the service provider is excluded. Hence, $\pi_2$ cannot be reconstructed by the players from the sequence $X^\prime$ as they cannot recover their blinded secret inputs from the ciphertexts $E_1(x_{\pi_2(i)}+r_{2_i})$ (see \reviewremove{the p} \reviewadd{P}roof of Lemma~\ref{Lem:PrivacyPlayers}). Therefore, the players cannot learn the permutation $\pi_2$, except with negligible probability.
	\\
	This completes the proof of Lemma~\ref{Lem:PermutationPrivacyPlayer}.%
\end{proof}

\begin{lemma}[Secret Permutation -- Service Provider]\label{Lem:PermutationPrivacyServiceProvider}
	The protocol SHUFFLE computes the shuffled sequence $\mathcal{X} = (...,E_{\reviewadd{1}}^\prime(x_{\pi(i)}),...)$ from the input sequence $X = (...,E_{\reviewadd{1}}(x_i),...)$ based on a random permutation $\pi$ such that the service provider cannot learn the permutation $\pi$ as long as there is no collusion between any player and the service provider.
\end{lemma}

\begin{proof}
	The ciphertext sent by $P_i$ in step~\ProtocolEncryptedRerandomisedBlindedInput{} is rerandomized (see \reviewremove{p} \reviewadd{P}roof of Lemma~\ref{Lem:SamePlaintexts}) and encrypts some player $P_{\rho_i}$'s secret input. It is chosen based on $P_i$'s uniformly distributed, unique index $\rho_i$. Therefore, the order of the encrypted (rerandomized) inputs in $\mathcal{X}$ computed in step~\ProtocolSequenceRerandomisedPermutedInputs{} is randomly permuted by the permutation $\pi_3$. The service provider cannot map the rerandomized ciphertexts of step~\ProtocolEncryptedRerandomisedBlindedInput{} to those of sequence $X^\prime$ with probability better than guessing as it does not learn the probabilistic encryptions of $0$, $E_1(0)$, and as the corresponding distributions are computationally indistinguishable (see Section~\ref{Preliminaries}). Therefore, an inversion of this permutation is only possible given the players' indices $\rho_i$, which are chosen uniformly at random (see \reviewremove{the p} \reviewadd{P}roof of Lemma~\ref{Lem:PermutedCiphertexts}). The service provider cannot decrypt the ciphertexts $E_1(r_{1_i})$ of step~\ProtocolEncryptedPlayersRandoms{} (see \reviewremove{p} \reviewadd{P}roof of Lemma~\ref{Lem:PrivacyServiceProvider}) and collusion between players and the service provider is excluded. Hence $P_S$ cannot compute the players' hashes of step~\ProtocolSequenceHashes{} and it cannot compute the players' random indices. Therefore, the order of the ciphertexts in $\mathcal{X}$ is randomly permuted by $\pi_3$, which can only be reconstructed by $P_S$ with negligible probability.
	\\
	This completes the proof of Lemma~\ref{Lem:PermutationPrivacyServiceProvider}.%
\end{proof}

As the final step of our proof of secrecy of the permutation $\pi$, we show that the proofs of Lemmas~\ref{Lem:PermutationPrivacyPlayer} and~\ref{Lem:PermutationPrivacyServiceProvider} are sufficient to prove that no participant can learn the overall permutation $\pi$ as long as there is no collusion between any player and the service provider, formalized in the following Lemma.

\begin{lemma}[Secret Permutation]\label{Lem:PermutationPrivacy}
	The protocol SHUFFLE computes the shuffled sequence $\mathcal{X} = (...,E_{\reviewadd{1}}^\prime(x_{\pi(i)}),...)$ from the input sequence $X = (...,E_{\reviewadd{1}}(x_i),...)$ based on a random permutation $\pi$. The permutation $\pi$ is not known to any participant as long as there is no collusion between any player and the service provider.
\end{lemma}

\begin{proof}
	To compute the composed permutation $\pi$, one needs to know both $\pi_2$ and $\pi_3$. As both independent permutations $\pi_2$ and $\pi_3$ are random, so is their composition $\pi$. To show that no participant can learn $\pi$, it is sufficient to show that no single participant can learn both $\pi_2$ and $\pi_3$. That is, for the composite permutation $\pi$ to be private, at least one of the two permutations $\pi_2$ or $\pi_3$ needs to be private for each participant.
	The proof of Lemma~\ref{Lem:PermutationPrivacyPlayer} proves that none of the $n$ players can learn $\pi$ as long as there is no collusion between any player and the service provider. The proof of Lemma~\ref{Lem:PermutationPrivacyServiceProvider} proves that the service provider cannot learn $\pi$ as long as there is no collusion between any player and the service provider.
	\\
	This concludes the proof of Theorem~\ref{Lem:PermutationPrivacy}.%
\end{proof}

\subsection{Summary}

As the final step of our proof of correctness, we show that the proofs of Lemmas~\ref{Lem:PermutedCiphertexts},~\ref{Lem:SamePlaintexts}, and~\ref{Lem:PermutationPrivacy} are sufficient to prove Theorem~\ref{Thrm:ProtocolIsSecureShuffle}, i.e.\reviewadd{,} Lemmas~\ref{Lem:PermutedCiphertexts},~\ref{Lem:SamePlaintexts}, and~\ref{Lem:PermutationPrivacy} imply Theorem~\ref{Thrm:ProtocolIsSecureShuffle}.

\begin{proof}
	From the proof of Lemma~\ref{Lem:PermutedCiphertexts}, it follows that the ciphertexts in sequence $\mathcal{X} = (...,E_{\reviewadd{1}}^\prime(x_{\pi(i)}),...)$ encrypting the plaintexts $x_i$ are permuted compared to the ciphertexts in sequence $X = (...,E_{\reviewadd{1}}(x_i),...)$ with a random permutation $\pi$. According to the proof of Lemma~\ref{Lem:PermutationPrivacy}, the permutation $\pi$ is not known to any participant as long as there is no collusion between any player and the service provider. From the proof of Lemma~\ref{Lem:SamePlaintexts}, it follows that the ciphertexts in sequence $\mathcal{X} = (...,E_{\reviewadd{1}}^\prime(x_{\pi(i)}),...)$ output by the protocol SHUFFLE encrypt the same plaintexts $x_i$ as the ciphertexts in the input sequence $X = (...,E_{\reviewadd{1}}(x_i),...)$. It also proves that the ciphertexts $E^\prime_1(x_i) = \mathcal{X}$ and $E_1(x_i) \in X$ encrypting the same plaintext $x_i$ are distinct, i.e.\reviewadd{,} $E_{\reviewadd{1}}^\prime(x_1) \neq E_{\reviewadd{1}}(x_1),...,E_{\reviewadd{1}}^\prime(x_n) \neq E_{\reviewadd{1}}(x_n)$. The combination of these proofs shows that the protocol SHUFFLE performs a secret shuffle of the $n$ players' encrypted inputs. For every sequence of ciphertexts $X = (...,E_{\reviewadd{1}}(x_i),...)$ with $1 \leq i \leq n$, the protocol SHUFFLE yields as output a sequence $\mathcal{X} = (...,E_{\reviewadd{1}}^\prime(x_{\pi(i)}),...)$ such that the ciphertexts $E_{\reviewadd{1}}^\prime(x_i) \neq E_{\reviewadd{1}}(x_i)$ have the same plaintexts $x_i$. Their order in $\mathcal{X}$ is randomly permuted by a permutation $\pi$. This permutation is not known to any participant as long as there is no collusion between any player and the service provider.
	\\
	This concludes the proof of Theorem~\ref{Thrm:ProtocolIsSecureShuffle}.%
\end{proof}
\section{Performance Evaluation}\label{Evaluationn}

\begin{figure*}[t!]
	\centering{
		\subfloat[Execution Time for $1024$ Bits]{\includegraphics[width=2.3in]{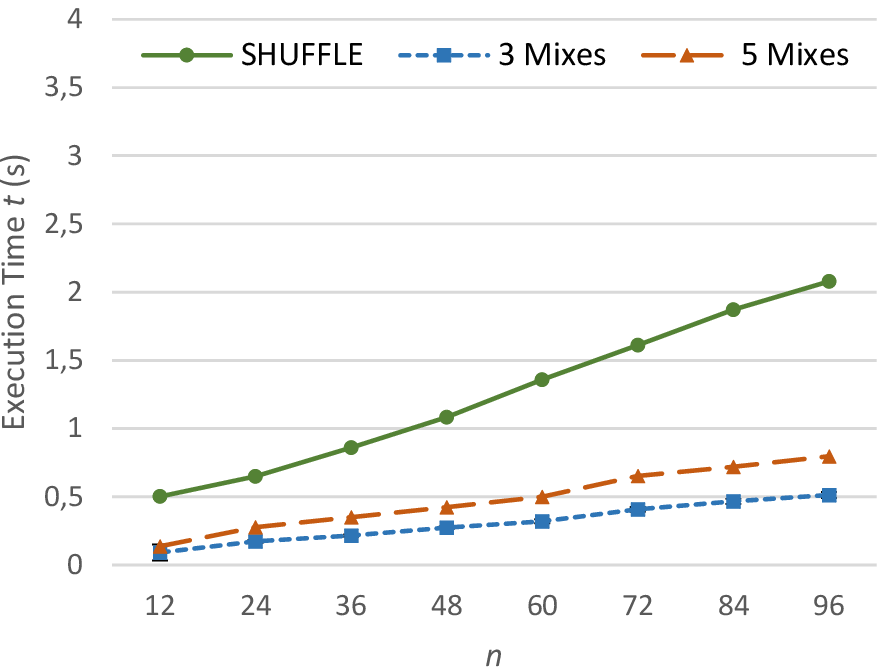}%
			\label{fig:EvaluationPeerGroupSize1024}}
		\hspace{0.15cm}
		\subfloat[Execution Time for $2048$ Bits]{\includegraphics[width=2.3in]{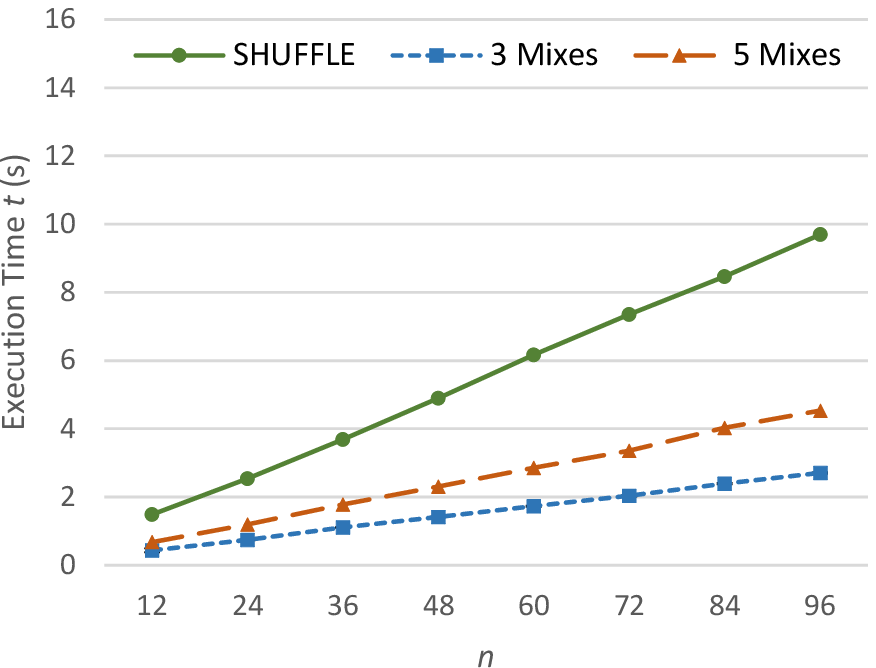}%
			\label{fig:EvaluationPeerGroupSize2048}}
		\hspace{0.15cm}
		\subfloat[Computation-to-Communication Ratio]{\includegraphics[width=2.3in]{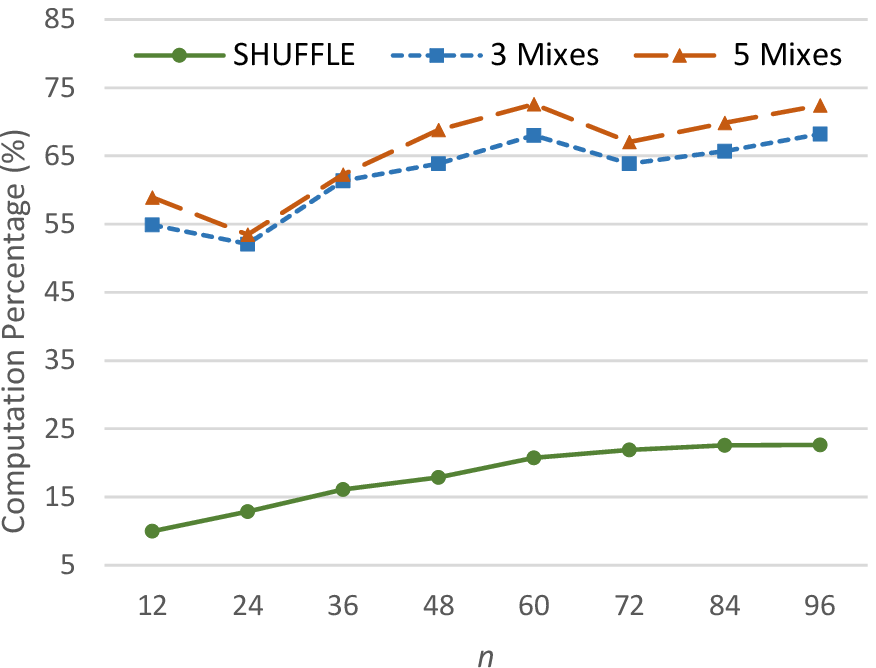}%
			\label{fig:EvaluationComputationPercentage}}}
	\caption{Results of the Empirical Performance Analysis}
\end{figure*}

The performance evaluation of our protocol SHUFFLE is twofold: We first investigate its asymptotic computational, communication, and round complexity in a theoretical analysis. Then, we examine its performance in an empirical analysis and compare it to the performance of mix networks.

\subsection{Asymptotic Complexity}\label{EvaluationAsmptoticComplexity}

\subsubsection{Round Complexity}

As depicted in Table~\ref{tab:DesignAdaptedProtocolTable}, the protocol consists of two rounds and a total of twelve protocol steps. Both values are independent of the number of players $n$. Therefore, the round complexity is constant in $n$, i.e.\reviewadd{,} $\mathcal{O}(1)$.

\subsubsection{Computational Complexity}\label{EvaluationAsmptoticComplexityComputation}

We investigate the number of operations that need to be carried out by the service provider and each player, respectively. We restrict our considerations to the cryptographic operations encryption, decryption, and ciphertext multiplication as they can be assumed to be the most complex ones. Their numbers are given in the middle columns of Table\reviewadd{s}~\ref{Tab:EvaluationAsmptoticComplexityService} and\reviewremove{ Table}~\ref{Tab:EvaluationAsmptoticComplexityPlayer}.
The resulting asymptotic computational complexity is $\mathcal{O}(n)$, i.e.\reviewadd{,} linear in the number of players $n$, for both the service provider and each player.

\begin{table}[t!]
	\renewcommand{\arraystretch}{1.4}
	\centering
	\caption{Service Provider's Computational and Communication Complexity of the Protocol}
	\begin{tabular}{c|ccc|c}
		\hline
		\textbf{Step} & \textbf{Enc} & \textbf{Dec} & \textbf{Mult} & \textbf{Message length}\\\hline\hline
		
		\textbf{\ProtocolSequencePermutedRandomValues{}} & & & & $n \cdot n \cdot l_{\mathcal{C}_1}$\\
		
		\textbf{\ProtocolSequenceEncryptedInputs{}} & $n$ & & $n$ & $n \cdot n \cdot l_{\mathcal{C}_1}$\\
		
		\textbf{\ProtocolSequenceEncryptedInputsBlindings{}} & $n$ & & & $n \cdot n \cdot l_{\mathcal{C}_2}$\\
		
		\textbf{\ProtocolServiceProvidersRandom{}} & & & & $n \cdot l_{\mathcal{M}_1}$ \\
		
		\textbf{\ProtocolSequenceRerandomisedPermutedInputs{}} & $n$ & $n$ & $2\cdot n$ & \\\hline
		
		\textbf{Total} & \textbf{$3\cdot n$} & \textbf{$n$} & \textbf{$3\cdot n$} & \textbf{$2 \cdot n^2 \cdot l_{\mathcal{C}_1} + n^2 \cdot l_{\mathcal{C}_2} + n \cdot l_{\mathcal{M}_1}$}\\\hline
	\end{tabular}
	\label{Tab:EvaluationAsmptoticComplexityService}
\end{table}

\begin{table}[t!]
	\renewcommand{\arraystretch}{1.3}
	\centering
	\caption{Each Player's Computational and Communication Complexity of the Protocol}
	\begin{tabular}{c|ccc|c}
		\hline
		\textbf{Step} & \textbf{Enc} & \textbf{Dec} & \textbf{Mult} & \textbf{Message length}\\\hline\hline
		
		\textbf{\ProtocolEncryptedInput{}} & $1$ & & & $l_{\mathcal{C}_1}$\\
		
		\textbf{\ProtocolEncryptedPlayersRandoms{}} & $1$ & & & $l_{\mathcal{C}_1}$\\
		
		\textbf{\ProtocolSequenceHashes{}} & & $n$ & & \\
		
		\textbf{\ProtocolRandomPositionHash{}} & & & & \\
		
		\textbf{\ProtocolEncryptedRerandomisedBlindedInput{}} & $1$ & & $1$ & $l_{\mathcal{C}_1}$\\
		
		\textbf{\ProtocolEncryptedBlindedBlinding{}} & $1$ & & $1$ & $l_{\mathcal{C}_2}$\\
		
		\textbf{\ProtocolEncryptedBlindingBlinding{}} & $1$ & & & $l_{\mathcal{C}_1}$\\\hline
		
		\textbf{Total} & \textbf{$5$} & \textbf{$n$} & \textbf{$2$} & \textbf{$4 \cdot l_{\mathcal{C}_1} + l_{\mathcal{C}_2}$}\\\hline
	\end{tabular}
	\label{Tab:EvaluationAsmptoticComplexityPlayer}
\end{table}

\subsubsection{Communication Complexity}\label{EvaluationAsmptoticComplexityCommunication}

To determine the communication complexity of the protocol, we investigate the length of the messages sent in each step of the protocol by the service provider and each player, respectively. These are given in the rightmost columns of Table\reviewadd{s}~\ref{Tab:EvaluationAsmptoticComplexityService} and\reviewremove{ Table}~\ref{Tab:EvaluationAsmptoticComplexityPlayer}. Here, $l_{\mathcal{M}_i}$ and $l_{\mathcal{C}_i}$ denote the maximum length of plaintexts in $\mathcal{M}_i$ and ciphertexts in $\mathcal{C}_i$, respectively. The total asymptotic communication complexity of each player is $\mathcal{O}(n)$, i.e.\reviewadd{,} linear in the number of players $n$. The service provider's communication complexity is $\mathcal{O}(n^2)$, i.e.\reviewadd{,} quadratic in the number of players $n$. Compared to related work, such as~\cite{Cha81}, our protocol has higher asymptotic communication complexity. However, we accept this loss as it helps reduce the computational complexity \reviewadd{asymptotically}.

\subsection{Empirical Performance}\label{EvaluationPerformance}

To investigate the practical performance of the protocol, we implemented both the players' and the service provider's part of the protocol and deployed them in a cloud-computing setting. The service provider was implemented as a Java HttpServlet and deployed in a cloud-computing instance with $96$ CPUs and $384$ GB RAM. To emulate sufficiently large numbers of independent players, we implemented the players' protocol steps in a Java HttpServlet and deployed the players in a Kubernetes cluster based on a cloud-computing instance with $96$ CPUs and $384$ GB RAM. We instantiated one Kubernetes node per player and provided each node with one CPU and $4$ GB RAM, which compares to the minimum requirements on a standard desktop computer. Therefore, we were able to emulate up to $96$ players. Service provider and players were deployed in different data centers in two major European cities with a distance of approximately $650$ kilometers to ensure a lifelike communication scenario. We used the additively homomorphic Paillier cryptosystem \reviewadd{for $\mathcal{CS}_1$ and $\mathcal{CS}_2$}.

For comparison, we implemented a simple yet efficient re-encryption mix network. Its construction is similar to the one described in~\cite{GJJS04}, but instead of the ElGamal cryptosystem with universal re-encryption, we used the standard version of Paillier's cryptosystem. Re-encryption (rerandomization) is performed given the public key of the players, which is a valid approach as the senders, i.e., players, in the shuffling scenario share the same key and the recipient, i.e., service provider, is not supposed to decrypt the received confidential data. We implemented the mixes as Java HttpServlets and deployed them in a similar cloud-computing setting as above, running each mix on an instance with $96$ CPUs and $384$ GB RAM. In a cascade of mixes, each mix receives all the messages at the same time in one batch, permutes and re-encrypts them, and forwards the full batch to the next mix or the recipient. This matches the communication setting of the service provider having the mix network shuffle all the messages once it received the full list from the players.

Fig.~\ref{fig:EvaluationPeerGroupSize1024} depicts the execution time $t$ relatively to the number of messages $n$ for $1024$-bit keys for our shuffling protocol and for mix networks with cascades of three and five mixes, respectively. Shuffling $96$ inputs with out shuffling protocol took $2.08$ seconds while the mix networks performed shuffling in $0.51$ and $0.80$ seconds, respectively. For $2048$-bit keys, shuffling $96$ inputs took $9.69$ seconds with our protocol and $2.71$ and $4.53$ seconds with mix networks (see Fig.~\ref{fig:EvaluationPeerGroupSize2048}). For both key lengths, the execution time of our shuffling protocol grows linearly in the number of players. Most importantly, the empirical results show that a mix network of five mixes with appropriate key length is only $2.14$ times faster than our shuffling protocol. However, recall that to achieve this performance, mix networks require multiple independent servers to perform the mixing whereas our shuffling protocol requires only a single server. Given the linear nature of re-encryption mix networks, one can reasonably assume that our protocol performs similar to a mix network of ten to eleven mixes. 

Furthermore, the linear growth of the execution time of our shuffling protocol indicates that the overall effect of the\reviewremove{ quadratic} communication complexity \reviewadd{being quadratic in $n$} is minor.
To further support this assumption, we investigated the ratio of computation time to communication time (see Fig.~\ref{fig:EvaluationComputationPercentage}). For growing $n$, the ratio of our protocol shows logarithmic trend. Besides that, its computation percentage is only a fraction of the computation percentage of the mix networks, which implies a smaller demand for computing power.

\subsection{Summary}

Our protocol has constant round complexity and linear computational complexity. Our empirical performance analysis shows that the execution time is linear in $n$. This implies that the \reviewremove{quadratic} \reviewadd{fact that the} communication complexity of our shuffling protocol \reviewadd{is quadratic in $n$,} only has a minor impact on the overall execution time. In this analysis, shuffling $96$ secret inputs encrypted under a $2048$ bits long Paillier key took $9.69$ seconds, which proves the practicability of our secret shuffling protocol. Performing shuffling via a mix network of five mixes takes roughly half as long. However, such a mix network requires five independent, powerful servers, each of which performs $\mathcal{O}(n)$ cryptographic operations and therefore consumes an amount of energy that is linear in $n$. Furthermore, the much smaller percentage of time required for computation relatively to communication in our shuffling protocol indicates generally lower cloud-computing costs and lower energy consumption.
Therefore, we consider our secret shuffling protocol as a valuable alternative to mix networks.
\section{Conclusion}\label{Conclusion}

We present an efficient secure multi-party protocol for shuffling encrypted data. It precludes any mapping between ciphertexts in the unshuffled and the shuffled sequence with probability better than guessing. We prove correctness of our shuffling functionality and privacy of the confidential inputs. Key element of our contribution is \reviewremove{the} \reviewadd{a} novel\reviewremove{,} \reviewadd{approach to} efficient random index \reviewremove{generation} \reviewadd{distribution}, which provides the random, secret permutation. The shuffling protocol has computational complexity linear in the number of players as well as constant round complexity. It shuffles $96$ ciphertexts in $9.69$ seconds for $2048$ bit long keys. We show that the effect of the communication complexity on the execution time is minor, which ensures good scalability. Our shuffling protocol performs asymptotically better than previous MPC-based shuffling approaches that focus on low communication complexity but suffer from higher computational complexity, which has negative impact on scalability. Furthermore, its execution time is only $2.14$ times that of a mix network of five mixes but requires no additional, independent servers. This not only enables use cases with centralized communication scenarios, but also causes much lower cloud-computing costs. \reviewadd{Being a general-purpose protocol, it can be used in a variety of applications such as privacy-preserving benchmarking systems, anonymous surveys, polls, voting, and many more.}
\section{Future Work}\label{FutureWork}

The protocol's applicability could be further improved by reducing its communication complexity. This can be achieved with a more efficient approach to obtaining the input ciphertexts from the service provider and selecting one of unique, random index. Moreover, it could be modified to be secure against malicious adversaries~\cite{Lin17}. In a more generic version, $m$ encrypted inputs could be present on the service-provider side prior to the protocol execution instead of being provided by the $n$ players. The $n$ players could then shuffle the $m$ values. Further security analysis is necessary to investigate the implications of setting $n \ll m$ where players generate multiple random indices and select and rerandomize multiple ciphertexts at once. This would further decrease the communication complexity and improve scalability of the protocol.

\printbibliography

\end{document}